\documentclass[orivec]{llncs}
\usepackage{graphicx}
\usepackage{latexsym}
\usepackage{algorithm}
\usepackage{algorithmic}
\usepackage{url}
\usepackage{verbatim}

\renewenvironment{proof}
{{\bf Proof:}}{\hspace*{\fill}$\Box$\par\vspace{2mm}}

\begin{document}
\title{On Touching Triangle Graphs}

\author{
Emden R.~Gansner\inst{1} \and Yifan Hu\inst{1} \and Stephen G.~Kobourov\inst{2}
}

\institute{
    AT\&T Labs - Research, Florham Park, NJ\\
    {\tt \{erg, yifanhu\}@research.att.com}
\and
    University of Arizona, Tucson, AZ\\
    {\tt kobourov@cs.arizona.edu}
}

\maketitle
\begin{abstract}
In this paper, we consider the problem of representing graphs by 
triangles whose sides touch. As a simple necessary condition, we
show that pairs of vertices must have a small common neighborhood. 
On the positive side, we present linear time algorithms for creating 
touching triangle representations for outerplanar graphs, square grid 
graphs, and hexagonal grid graphs. 
We note that this class of graphs is not
closed under minors, making characterization difficult. 
However, we present a complete characterization of the subclass of 
biconnected graphs that can be represented as triangulations of some polygon. 
\end{abstract}

\section{Introduction} Planar graphs are a widely studied class of graphs
that includes naturally occurring subclasses such as trees and outerplanar
graphs. Typically planar graphs are drawn using the node-link model, where
vertices are represented by a point and edges are represented by
line segments. Alternative representations, such as contact
circles~\cite{Brightwell:1993:RPG} and contact
triangles~\cite{contact_triangles} have also been explored. In these
representations, a vertex is a circle or triangle, and an edge is represented by
pairwise contact at a common point.

In this paper, we explore the case where vertices are polygons, with an edge
whenever the sides of two polygons touch.
Specifically, given a planar graph $G=(V,E)$,
we would like to find a set of polygons $R$ such that:
\begin{enumerate} 
\item there is bijection between $V$ and $R$; 
\item two polygons touch non-trivially if and only if the corresponding vertices are adjacent in $G$; 
\item and each polygon is convex.
\end{enumerate}

Note that, unlike the case of contact circle and contact triangle representations,
two polygons that share a common point are not considered adjacent. It
is easy to see that all planar graphs have representations meeting conditions
1 and 2 above, as pointed out by de Fraysseix {\em et al.}~\cite{FMR04}.
Starting with a straight-line planar drawing of $G$, a polygon for each vertex
can be defined by taking the midpoints of all adjacent edges and the centers
of all neighboring faces. The {\em complexity}, i.e., 
number of its sides, of the resulting polygons can be as
high as $|V|-1$, as it is proportional to the degree of the corresponding
vertex. Moreover, the polygons would not necessarily be convex. 

A theorem of Thomassen \cite{rl:t84} implies that all planar graphs can be represented
using convex hexagons. (This also follows from results by Kant \cite{kant-92} and
de Fraysseix {\em et al.}~\cite{FMR04}.) Gansner {\em et al.}~\cite{ghkk09} have
shown that six sides are also necessary.
This leads us to consider which planar graphs can be represented by polygons with
fewer than six sides.

This paper presents some initial results for the case of touching triangles.
We assume we are dealing with connected planar graphs $G=(V,E)$. We let 
$TTG$ denote the class of graphs that have a touching triangle representation.
In Section~\ref{sec:outerplanar}, we show that all outerplanar graphs are in $TTG$.
Similarly, we show in Section~\ref{sec:grids} that all subgraphs of
a square or hexagonal grid are in $TTG$. All of these representations can be computed
in linear time. Section~\ref{sec:triangulations} characterizes the special case of graphs
arising from triangulations of simple, hole-less polygons. Finally, 
in Section~\ref{sec:necessary}, we show that, for graphs in $TTG$, 
pairs of vertices have very limited common neighborhoods. 
This allows us to identify concrete examples of graphs not in $TTG$.

\subsection{Related Work}

Results on representing planar graphs as ``contact systems'' can be dated
back to Koebe's 1936 theorem~\cite{Koebe36} which states that any planar
graph can be represented as a contact graph of
disks in the plane. When the regions are further restricted to rectangles,
not all planar graph can be represented. Rahman {\em et al.}~\cite{Rahman04}
describe a linear time algorithm for constructing rectangular contact graphs,
if one exists. Buchsbaum {\em et al.}~\cite{Buchsbaum08} provide a
characterization of the class of graphs that admit rectangular contact graph
representation. The version of the problem where it is further required that
there are no holes in the rectangular contact graph representation is known
as the rectangular dual problem. He~\cite{He93} describes a linear time
algorithm for constructing a rectangular dual of a planar graph, if one
exists. Kant's linear time algorithm for drawing degree-3 planar graphs on a
hexagonal grid~\cite{kant-92} can be used to obtain hexagonal drawings for
planar graphs.

In VLSI floor-planning it is often required to partition a rectangle into
rectilinear regions so that non-trivial region adjacencies correspond to a
given planar graph. It is natural to try to minimize the complexities of the
resulting regions and the best known results are due to He~\cite{He99} and
Liao {\em et al.}~\cite{Liao03} who show that regions need not have more than
8 sides. Both of these algorithms run in $O(n)$ time and produce layouts on
an integer grid of size $O(n) \times O(n)$, where $n$ is the number of
vertices.

Rectilinear cartograms can be defined as rectilinear contact graphs for
vertex weighted planar graphs, where the area of a rectilinear region must be
proportional to the weight of its corresponding node. Even with this extra
condition, de Berg {\em et al.}~\cite{deBerg07} show that rectilinear
cartograms with constant region complexity can be constructed in $O(n \log
n)$ time. Specifically, a rectilinear cartogram with region complexity 40 can
always be found.

\section{Outerplanar Graphs} \label{sec:outerplanar}

In this section, we show that any outerplanar graph can be represented by a
set of touching triangles, that is, outerplanar graphs belong to the class
$TTG$. Here we assume that we are given an outerplanar graph $G=(V,E)$ and
the goal is to represent $G$ as a set of touching triangles. We describe a
linear time algorithm based on inserting the vertices of $G$ is an
easy-to-compute ``peeling'' order.
Figure~\ref{fig-outerplanar} illustrates the algorithm with an
example.

\begin{figure}[th] 
\begin{center}
\vspace{-1cm}
\includegraphics[width=12cm]{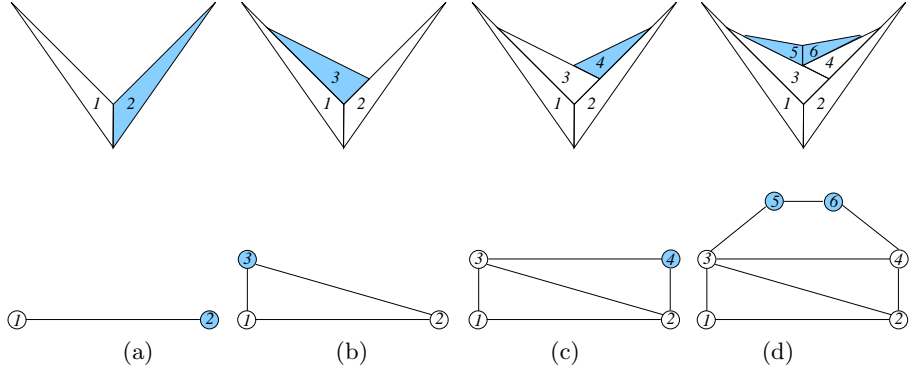}\\ 
(a) \hspace{2.2cm} (b) \hspace{2.2cm} (c) \hspace{2.2cm} (d) \caption{\small\sf
Incremental construction of the $TTG$ representation of an outerplanar graph on
6 vertices. The shaded vertices on the bottom row and shaded regions on the
top row are the ones processed at the current step. Note that the upper
envelope on the top row forms a concave chain at all times.
\label{fig-outerplanar} 
} 
\end{center} 
\end{figure}

\subsection{Algorithm Overview}

\begin{enumerate} 
\item Compute an outerplanar embedding of $G$. 
\item Compute a reverse ``peeling'' order of the vertices of $G$. 
\item Insert region(s) corresponding to the current set of vertices 
in the peeling order, while maintaining a concave upper envelope. 
\end{enumerate} 

We now look at each step in more detail.
The first step of the algorithm is to compute an outerplanar embedding of the
graph, that is, an embedding in which all the vertices are on the outer face.
For a given planar graph $G=(V,E)$, this can be easily done in linear
time as follows. Let $w$ be a new vertex and let $G'=(V',E')$, where
$V'=V\cup\{w\}$ and $E'=E\cup\{(v,w)$ for all $v\in V\}$. Note that 
$G'$ is planar: if it contained a subgraph homeomorphic to $K_5$ or 
$K_{3,3}$, then $G$ would contain a subgraph homeomorphic to $K_4$ or
$K_{3,2}$, which would imply that G was not outerplanar to begin with 
as these are forbidden graphs for outerplanar 
graphs (Theorem 11.10, \cite{h-gt-72}).
We can then
compute a planar embedding for $G'$ with $w$ on the outer face. Removing
$w$ and all its edges yields the desired outerplanar embedding.

The second step of the algorithm is to compute a reverse ``peeling'' order of
the vertices of $G$. Such an order is defined by peeling off one face at a
time and keeping track of the set of removed vertices. Note that, as $G$ is
outerplanar, each such set is a path with one or more vertices and only its
endpoints are connected to the rest of the graph. Moreover, as the dual of an
outerplanar graph is a tree, any pair of adjacent faces shares exactly one
edge. As a result of this step in the algorithm, all the vertices of $G$ are
partitioned into disjoint sets with increasing labels. Since the order is
reversed, the last face peeled is the one with vertices $v_1, v_2, v_3$.

The third step of the algorithm is to create the touching triangles
representation of $G$, by processing the graph using the peeling order from
the second step. We begin by placing the vertices in the last peeled face.
Suppose the last peeled face has exactly 3 vertices, $v_1, v_2, v_3$. Without
loss of generality, let the edge $(v_2,v_3)$ separate this face from the rest
of the graph. We create two triangles corresponding to $v_1$ and $v_2$ and
place these triangles so that they have one adjacent side and two other sides
of the triangles create a concave upper envelope; see
Fig.~\ref{fig-outerplanar}(a). The third vertex, $v_3$, corresponds to a
triangle that can be placed in the created concavity so that it has one side
touching the triangle that corresponds to $v_1$ and another side touching the
triangle that corresponds to $v_2$. The size of the triangle is computed so
that the upper envelope is still concave and contains a side of each of the
three triangles; see Fig.~\ref{fig-outerplanar}(b). Taking the midpoints of
the adjacent sides of the already placed triangles for $v_1$ and $v_2$ would
do.

In general, when processing the current set of one or more vertices in the
peeling order, they are of the form $v_k, v_{k+1}, \dots, v_{k+j}$, $j\geq
0$. These vertices form a path in $G$ and $v_{k+1}, \dots, v_{k+j-1}$ each
have degree 2 in the current graph, that is, they are not connected to any
other vertices of the graph processed so far, due to outerplanarity.
Furthermore, $v_k$ and $v_{k+j}$ are connected to two other vertices in $G$
which have already been processed; call them $v_l$ and $v_r$. Due 
to outerplanarity, $v_l$ and $v_r$ correspond to two adjacent triangles in the
concave upper envelope. If $j=0$, we just need to create one triangle that
corresponds to the single current vertex $v_k$ and place it so that it is
adjacent to the already processed triangles corresponding to $v_l$ and $v_r$,
and ensuring that the new triangle preserves the concavity of the upper
envelope. Once again, taking the midpoints of the adjacent sides of the
already placed triangles for $v_l$ and $v_r$ suffices; see
Fig.~\ref{fig-outerplanar}(c).

If $j>0$, then we represent the $j+1$ current vertices as a ``fan'' of
triangles that have adjacent sides and are also adjacent to the two already
placed triangles that correspond to $v_l$ and $v_r$. Finally, we ensure that
the upper envelope of the resulting group of triangles forms a concave
envelope; see Figure~\ref{fig-outerplanar}(d). Note that this idea can be
applied to the case when the first peeled face is made of more than 3
vertices.

The algorithm maintains the following two invariants:
\begin{enumerate} 
\item the upper envelope of the touching-triangles representation is concave. 
\item all vertices that might still have incoming
edges in a future stage of the algorithm have an exposed side in their
corresponding triangle on the upper envelope. 
\end{enumerate}

The first step of this algorithm can be done in linear time as it is a
slight modification of a standard planar embedding algorithm such as that by
Hopcroft and Tarjan~\cite{ht-ept-74}. The second step can also be done in
linear time as computing the ``peeling ordering'' requires constant time per
face, given the embedding of the graph from the previous step. In the third
step, we record the three edges of each triangle corresponding to each
processed vertex. Inserting a new chain of vertices involves finding the 
midpoint of the exposed edges, and forming the ``fan'' of new triangles, all
tasks which require constant time per vertex and add up to linear overall
time. Thus, we have the following theorem:

\begin{theorem} 
A touching triangles representation can be computed in linear
time for any outerplanar graph. 
\end{theorem}

\section{Grid Graphs} 
\label{sec:grids} 
In this section, we show that any subgraph of a
square or hexagonal grid graph can be represented by a set of touching
triangles. We describe a linear time
algorithm based on inserting the vertices of the graph in an outward fashion
starting from an interior square/hexagon. 
We illustrate the algorithm with examples in Figure~\ref{fig-grids}.

\begin{figure}[th] 
\begin{center}
\vspace{-1cm}
\includegraphics[width=5cm]{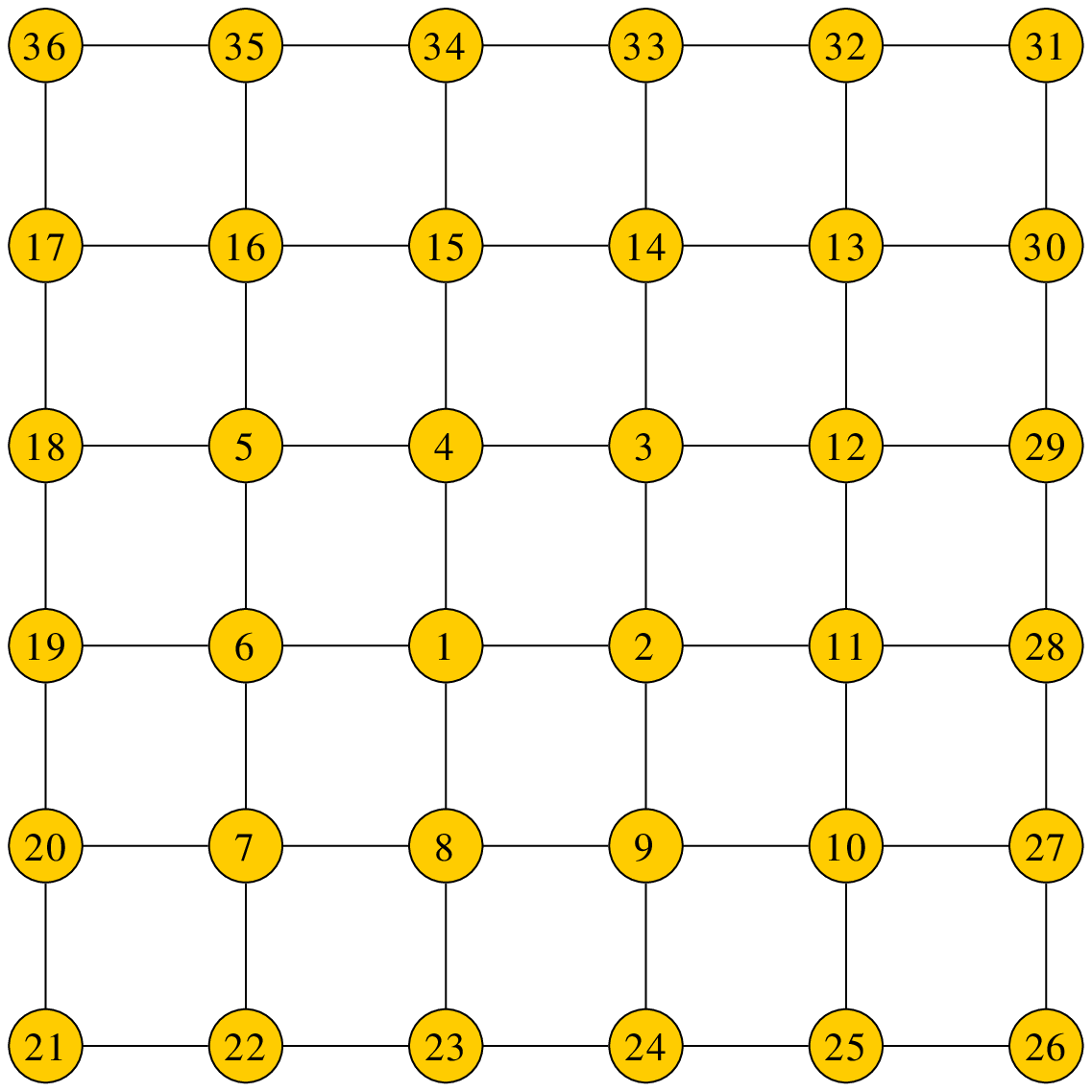}\hspace{1cm}
\includegraphics[width=6cm]{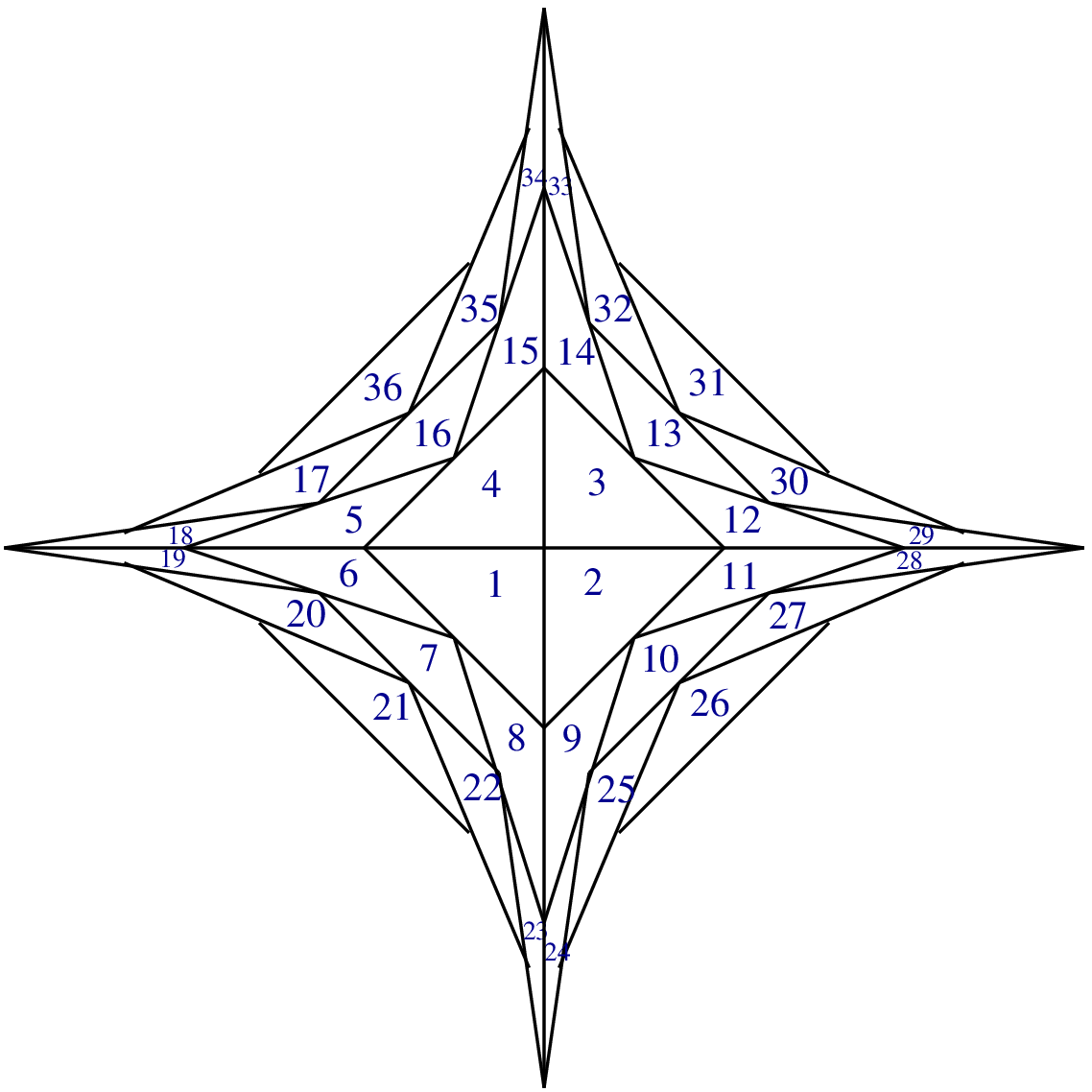}
\includegraphics[width=4.8cm]{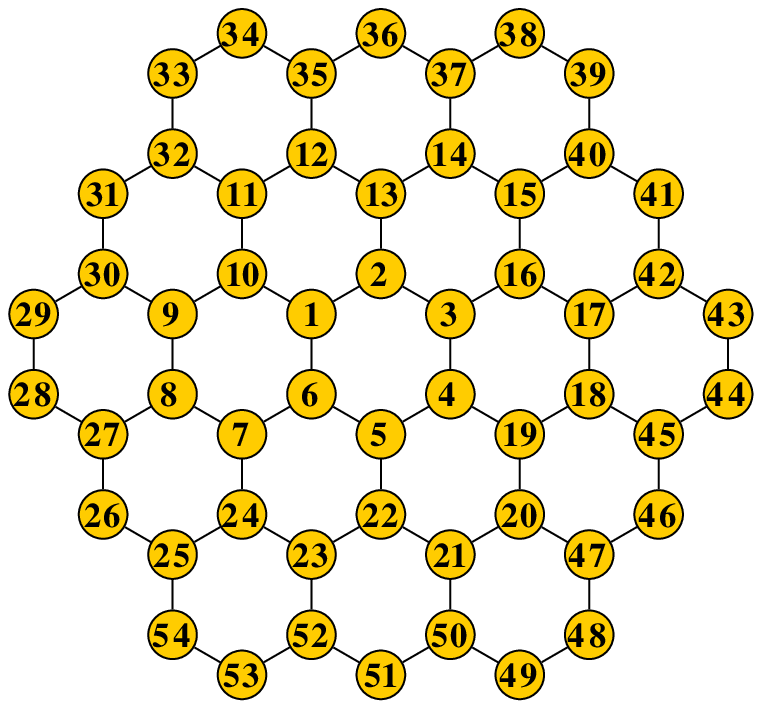}\hspace{1cm}
\includegraphics[width=4.9cm]{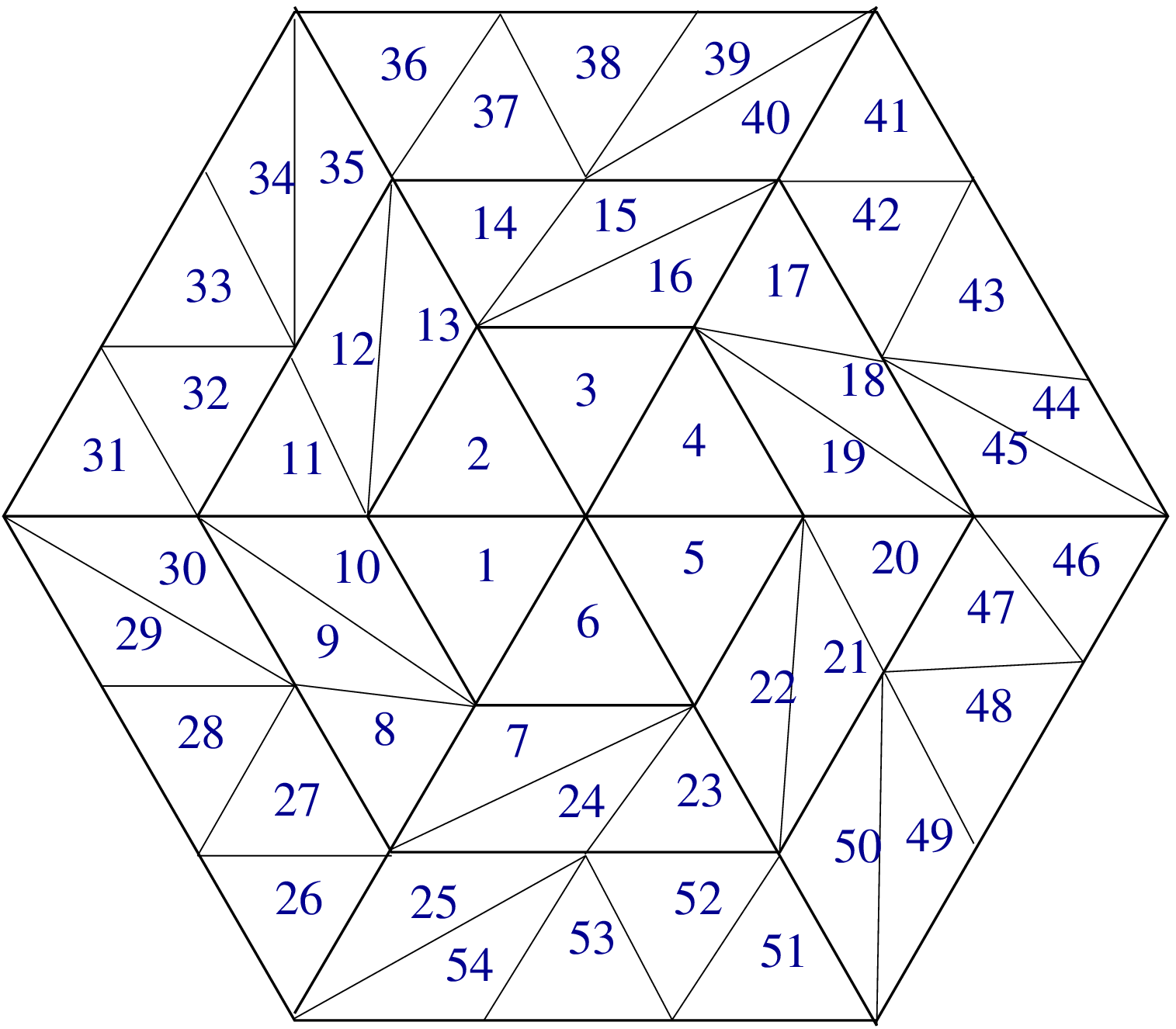} 
\caption{\small\sf Grid
graphs (left) as touching triangles (right). The Hamiltonian path that visits
all the vertices in the spiral order is given by the labels of the vertices.}
\label{fig-grids} 
\end{center} 
\end{figure}

\subsection{Algorithm Overview}
We first consider $TTG$ representations for grid graphs.
\begin{enumerate} 
\item Compute a planar embedding of $G$. 
\item Compute a ``spiral'' order of the vertices of $G$. 
\item Insert region(s), corresponding to a vertex or a path of 
vertices in the spiral order, while
maintaining a concave upper envelope in each quadrant (in the case of square
grid), or by carving out triangles out of trapezoids that correspond to the
current spiral segment (in the case of hexagonal grid). 
\end{enumerate} 

The first step of the algorithm is to compute a planar embedding of the
graph, which can be done in linear time~\cite{ht-ept-74}. Next we compute a
``spiral'' order of the vertices. Such an order is defined by a Hamiltonian
path which starts with the innermost face and visits all the vertices as
shown in Fig.~\ref{fig-grids}. Note that this is well defined for symmetric
grid graphs but can be modified to handle asymmetric grid graphs and
subgraphs of grid graphs.

In the case of square grids, the plane is partitioned into four quadrants and
in each quadrant the spiral order introduces vertices in paths of increasing
lengths $(1, 3, 5, \dots)$. In general these paths can be introduced
recursively, provided that the upper envelope of the quadrant remains
concave. The insertion of regions is similar to the process described for
outerplanar graphs above.

In the case of hexagonal grids the plane is partitioned into six sectors and
in each sector the spiral order introduces vertices in paths of increasing
lengths $(1,3,5, \dots)$. In general, these paths can be introduced directly
by adding an adjacent trapezoidal region and carving it into triangles. 

The above algorithms show how to construct a $TTG$ representation for
any square or hexagonal grid graph. To get a $TTG$ representation
for any subgraph, one need only remove the triangles corresponding to
vertices unused in the subgraph, and adjust the remaining triangles to
remove any contacts corresponding to unused edges.
Thus, we have the following theorem:

\begin{theorem} 
A touching triangles representation can be computed in linear
time for any subgraph of a square or hexagonal grid graph. 
\end{theorem}

\section{Triangulations}
\label{sec:triangulations}

If we require each face in a triangle representation to have
exactly three vertices, i.e., the vertex of one triangle cannot
touch the side of another, we get the 
special case of $TTG$s we call {\em triangulation graphs}.
These representations clearly correspond to creating a triangular
mesh \cite{bern,cg-book-00}, allowing Steiner points, 
within the interior of a polygon. For example, the representation in
the bottom right of Fig.~\ref{fig-grids} is a triangulation graph and the
representation in the top right of Fig.~\ref{fig-grids} is not.

It is easy to see that triangulation graphs form a strict subset of
$TTG$s. For example, $K_4$ is a $TTG$ but not a triangulation graph.
It is also immediate that a triangulation graph has maximum degree 3,
because by the definition of triangulation graphs, the vertex of 
one triangle cannot touch the side of another.

\begin{lemma}
\label{prop:deg2}
If G is a triangulation graph with no nodes of degree 1, 
G has at least 3 nodes of degree 2.
\end{lemma}

\begin{proof}
The only triangles that can contribute to the polygon's boundary
or outer face
must have degree 2 in the graph, each contributing exactly 1 edge to
the boundary. Since the polygon has at least 3 edges, the result follows.
\end{proof}

A further subclass consists of the {\em filled triangulation graphs}, those
who have a representation whose corresponding polygon is simple with no holes. 
It is possible to fully characterize the biconnected subset of these graphs.

\begin{theorem}
\label{thm:triang}
Assume $G$ is biconnected.
$G$ is a filled triangulation graph if and only if $G$ has:
\begin{enumerate}
\item only nodes of degree 2 or 3
\item an embedding in the plane such that:
\begin{enumerate}
\item every internal node has degree 3;
\item there are at least 3 nodes of degree 2 on the boundary;
\item if there are any degree 3 nodes on the boundary, all of the degree 2 nodes
cannot be consecutive; and
\item if the degree 2 nodes on both ends of a chain of degree 3 boundary
nodes are removed, the graph remains connected.
\end{enumerate}
\end{enumerate}
\end{theorem}

\begin{proof}
Let $G$ be a filled triangulation graph. Since it is biconnected,
it cannot have any vertices of degree 1. Its triangulation representation
yields an embedding with all internal nodes of degree 3. 
Lemma~\ref{prop:deg2} shows we have at least 3 nodes of degree 2
on the boundary.

Suppose there are degree 3 nodes on the boundary
and the degree 2 nodes are consecutive. The chain of degree 2 nodes
cannot connect at a single vertex, because this would be a cut vertex.
Thus, if we remove all triangles corresponding to degree 2 nodes, we
would have a triangulation representation of a graph with exactly
2 vertices of degree two, which is not allowed by Lemma~\ref{prop:deg2}.

To finish the proof of necessity, we note that for two degree 2 triangles
to disconnect the triangulation, they would have to share an interior vertex.
On the other hand, if all intervening triangles on the boundary have
degree 3, they can contribute nothing to the polygon boundary, so the
two degree 2 must share another vertex. But then, they share a side,
so there can't be any intervening degree 3 triangles.

Next, we prove sufficiency. We assume $G$ is biconnected, all of its
vertices have degree 2 or 3, and it has the specified embedding.
We construct a graph $G'$ which is a special kind of dual of $G$. $G'$ contains
the dual of the interior faces and edges of $G$. In addition, $G'$ has a vertex
for each maximal sequence of degree 3 nodes on the boundary, and
a vertex for each boundary edge connecting two degree 2 nodes. These are placed
in the external face of $G$, near the corresponding nodes or edges. These
vertices are connected in a cycle of $G'$ following the ordering induced by
the boundary nodes and edges of $G$. Finally, for each boundary edge 
$e$ of $G$,
we add an edge from the node of $G'$ corresponding to the interior face of
$G$ containing $e$ to one of the vertices on the external cycle of $G'$. 
If $e$ is adjacent to a vertex of degree 3, we connect the edge to the
node of $G'$ corresponding to the degree 3 vertex. Otherwise, we connect
to the node of $G'$ corresponding to $e$.

It is immediate from the construction that $G'$ is a planar embedding of
nodes and edges; all interior faces are triangles; 
and there is a 1-1 correspondence between faces of $G'$ and vertices of $G$ 
and between edges in $G$ and $G'$. We need to show that $G'$ is a simple graph.

As $G$ is biconnected, $G'$ can have no loops. Property 2(d) of the 
embedding implies that each interior face is connected to at most one of
the nodes associated with the exterior face.
The only way that multiedges could then occur would be if 
$G'$ has a boundary consisting of
two nodes and two edges. We know $G$ has as least $n_2 \ge 3$ nodes of 
degree 2 on the boundary. If there are only degree 2 nodes on the boundary,
$G'$ has a boundary of $n_2$ nodes. Assume $G$ has some degree 3 nodes on the
boundary. If these nodes split into 3 or more paths, the construction creates
at least 3 nodes on the boundary of $G'$. If not, they must split into 2 paths,
since the degree 2 nodes must be separated. One group of degree 2 nodes must
contain at least 2 nodes. The construction then creates one node for each
group of degree 3 nodes, and at least one node for the path of more
than 2 degree nodes, again given $G'$ at least 3 boundary nodes.

As $G'$ is simple, by using one
of the algorithms (e.g, \cite{fpp-sssfe-88}) for making the edges of 
planar graph into line segments while retaining the embedding, we derive a 
triangulation representation of $G$, completing the proof.
\end{proof}

\begin{figure}[th]
\centering
\vspace{-1cm}
\begin{tabular}{ccc}
\includegraphics[scale=0.35]{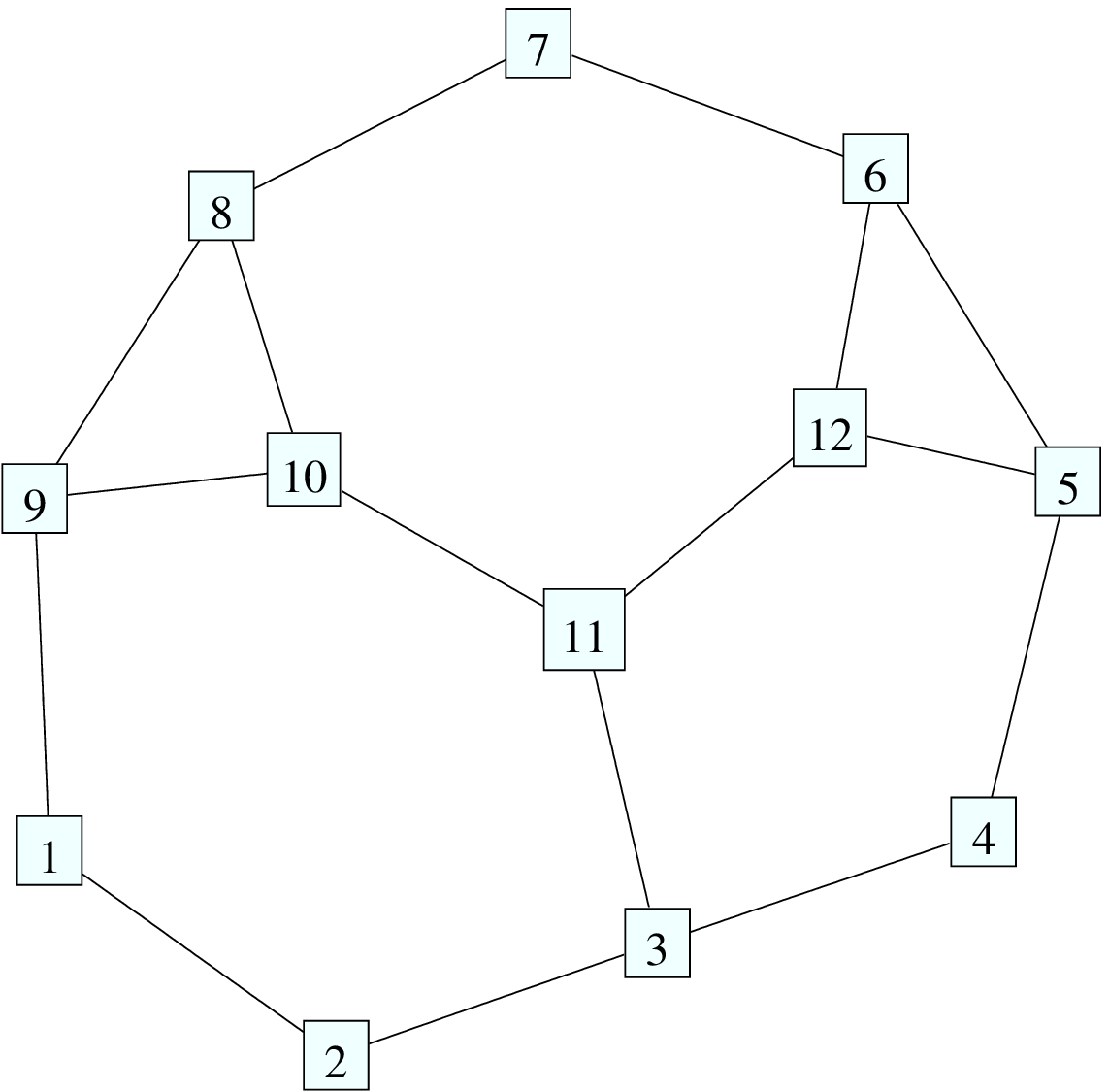} & \includegraphics[scale=0.35]{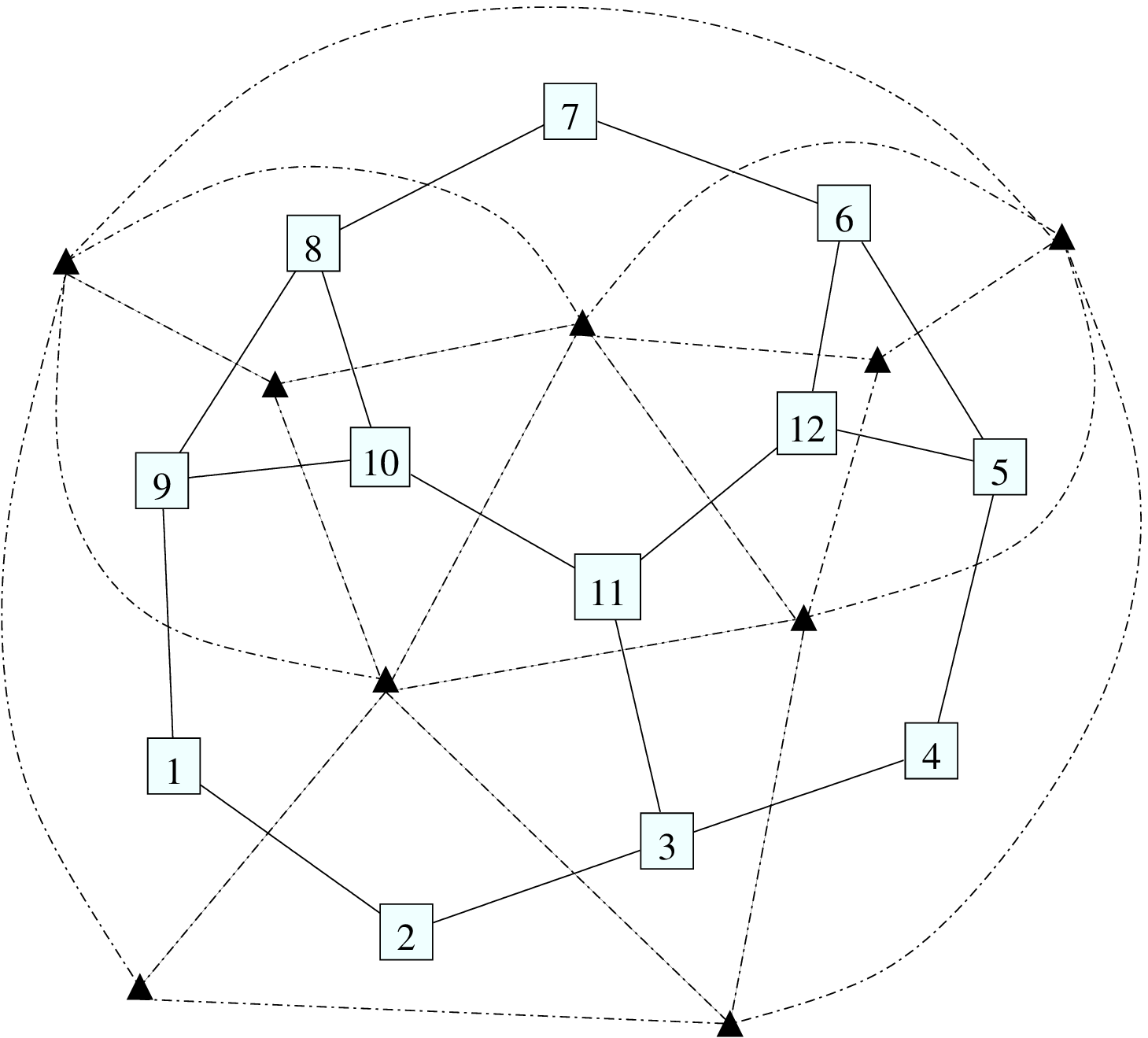} & \includegraphics[scale=0.35]{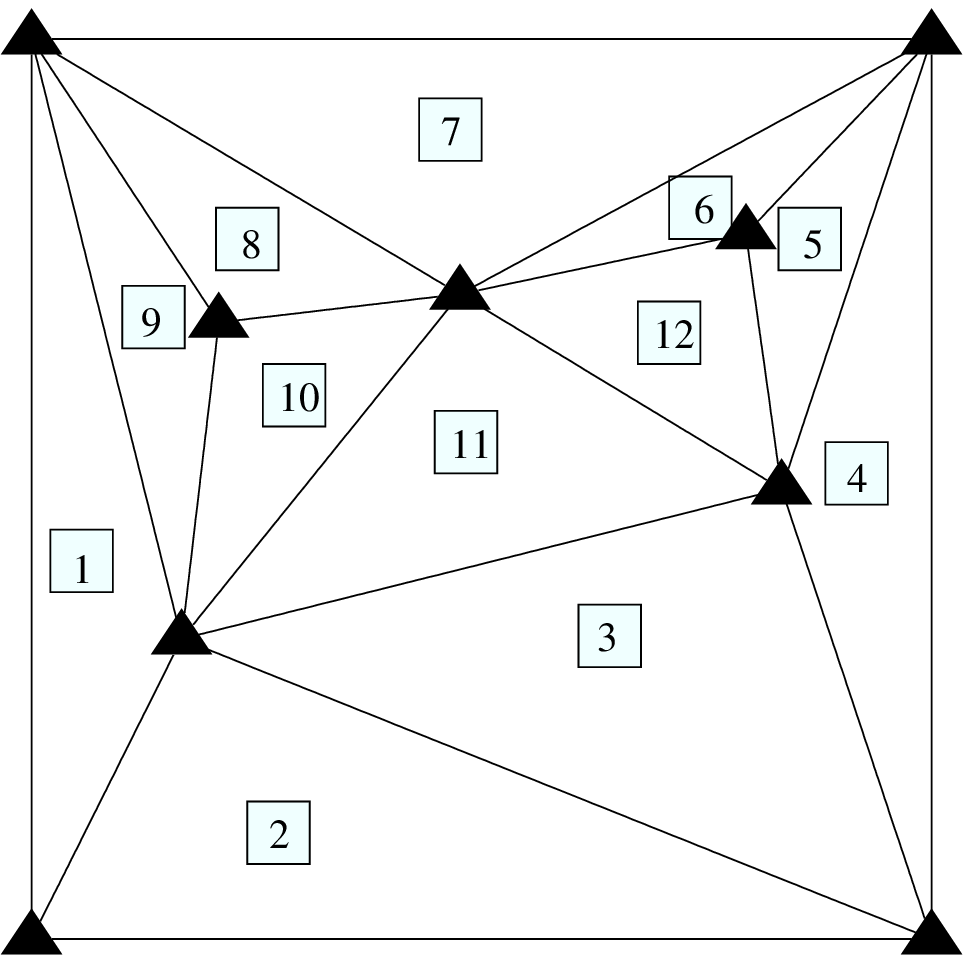} \\
{\bf (a)} & {\bf (b)} &  {\bf (c)}
\end{tabular}
\caption{\small\sf Constructing a triangulation graph. {\bf (a)} Original graph; {\bf (b)} Creating the ``dual'' graph; {\bf (c)} Straightening the edges.}
\label{fig:triang}
\end{figure}

Perhaps not surprisingly, the conditions of the theorem have a similar
feel to those for rectangular drawings \cite{Rahman04}.
It is also not hard to see that the result can probably be derived from
the duality between planar, cubic, 3-connected graphs and triangulations
of the plane \cite{SR2}, but our proof seems more straightforward. 
Lastly, we note that Theorem~\ref{thm:triang} gives another proof
that the hexagonal grid graphs of Section~\ref{sec:grids} have a touching
triangle representation.

Figure~\ref{fig:triang} demonstrates the algorithm. 
Figure~\ref{fig:triang}(a) shows a graph satisfying the conditions of the
theorem. In Figure~\ref{fig:triang}(b), we have added a node for each
internal face, and node on the outside for each sequence of degree 3 nodes
or for each edge both of whose nodes have degree 2. This gives us a planar
graph with each face having three sides and associated with a node of the
original graph. Straightening the sides of the faces makes each face
a triangle.

\section{Necessary conditions}
\label{sec:necessary}

Thus far, we have shown that various categories of graphs are in 
$TTG$. Now, we wish to pursue some necessary conditions which will
eliminate many graphs from $TTG$. We start with some definitions.

Given triangles $T_0$ and $T_1$, pick
two sides $s_0$ and $s_1$, one from each triangle,
and orient the side counter-clockwise around the interior of
the triangle. Extend the
sides into directed lines $L_0$ and $L_1$.
If the lines intersect at a unique point, the intersection
is {\em feasible} if a non-trivial portion of $s_0$ lies
to the right of $L_1$ and a non-trivial portion of
 $s_1$ lies to the right of $L_0$. Of the four
angles formed at a feasible intersection, there is a unique one
corresponding to a right turn. We call this a {\em feasible angle}.
Two sides are {\em collinear} if the directed lines $L_0$ and $L_1$
are identical.

\begin{lemma}If a triangle T touches both $T_0$ and $T_1$, using two
distinct sides, one of its angles must be a feasible angle of T0 and T1. 
\end{lemma}

\begin{proof}
If $\alpha$ is the angle of $T$ determined by the two touching sides
of $T_0$ and $T_1$, it immediate that  $\alpha$ is a feasible angle.
See Figure~\ref{fig:feasible}.
\end{proof}

\begin{figure}[th]
\begin{center}
\vspace{-.5cm}
\includegraphics[scale=.5]{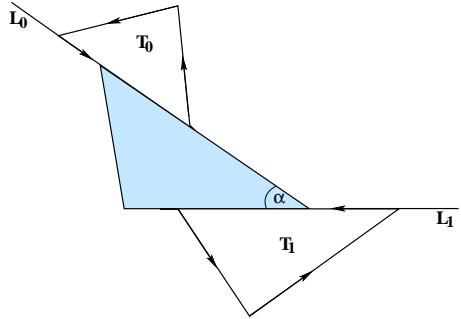}
\caption{\small\sf A triangle $T$ touching two other triangles 
$T_0$ and $T_1$. The angle $\alpha$ is a feasible angle of
$T_0$ and $T_1$.}
\label{fig:feasible}
\end{center}
\end{figure}

This lemma already greatly reduces the possible $TTG$ graphs.
If two triangles have no collinear sides, there can be at most nine
triangles touching both both of them, since any such triangle
eliminates at least one of the feasible angles.
If two sides are collinear, one triangle can touch those two sides.
Any other triangles must correspond to feasible angles, and since
the remaining sides of both triangles are all to the left of the two
collinear sides, there can be at most 4 feasible angles.
We next work at tightening these bounds.

For a node $u$ in $G$, we let $N_u$ be the nodes in $G$ joined to
$u$ by an edge.
If $u$ and $v$ are two nodes in a graph $G$, define $N_{uv}$ as the
mutual neighbors of $u$ and $v$, that is, $N_{uv} = N_u \cap N_v$.
Finally, define $E_{uv}$ be the subset of edges of $G$ induced by
$N_{uv}$.

\begin{theorem}\label{thm:sideside}
Let G be a $TTG$, and let $u$ and $v$ be two nodes
in G joined by an edge. Then $|N_{uv}| \le 3$ and $|E_{uv}| \le 1$.
\end{theorem}
\begin{proof}
Let $T_u$ and $T_v$ be the two triangles corresponding to nodes $u$ and $v$.
Since the two nodes share an edge,  $T_u$ and $T_v$ must touch. There are basically 
two possibilities: one side is totally contained in the other or not.

In the first case, we have the situation represented in Figure~\ref{fig:inside}.
We immediately note that there can be no feasible angle associated with $\vec{12}$ and
$\vec{ab}$. In addition, $\vec{ab}$ is to the left of both $\vec{23}$ and $\vec{31}$.
On the other hand, there are feasible angles formed by $\vec{12}$ with $\vec{bc}$ and
$\vec{ca}$. So, we only have to consider pairings of $\vec{23}$ and $\vec{31}$
with $\vec{bc}$ and $\vec{ca}$.

\begin{figure}[ht]
\centering
\begin{tabular}{ccc}
\includegraphics[scale=0.5]{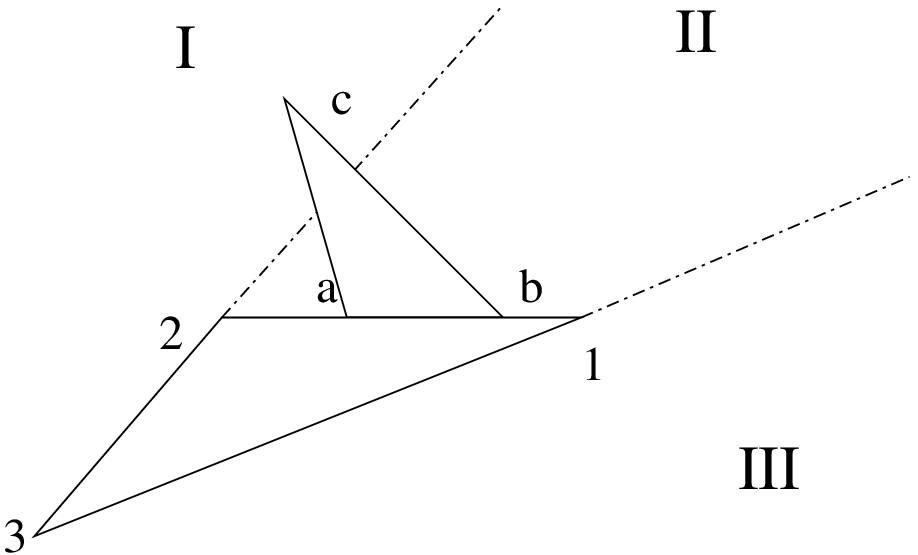} & \includegraphics[scale=0.5]{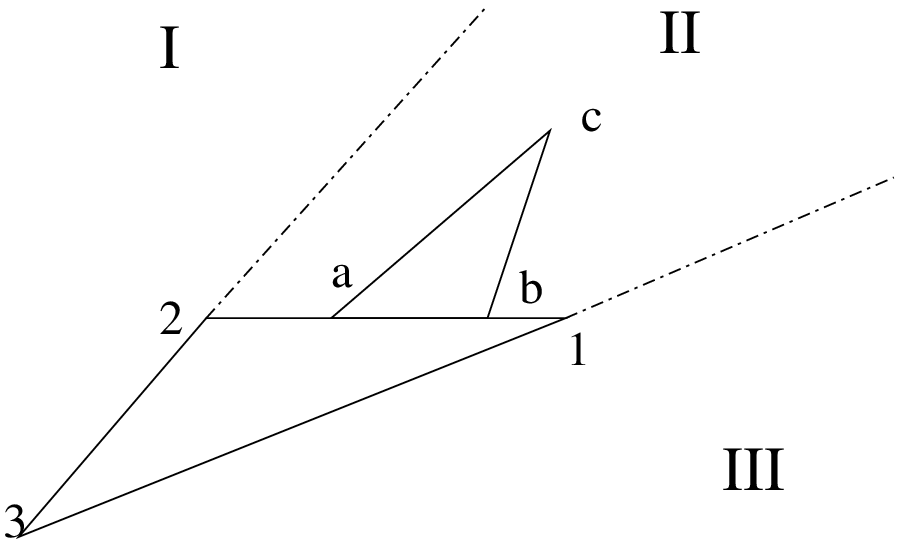} & \includegraphics[scale=0.5]{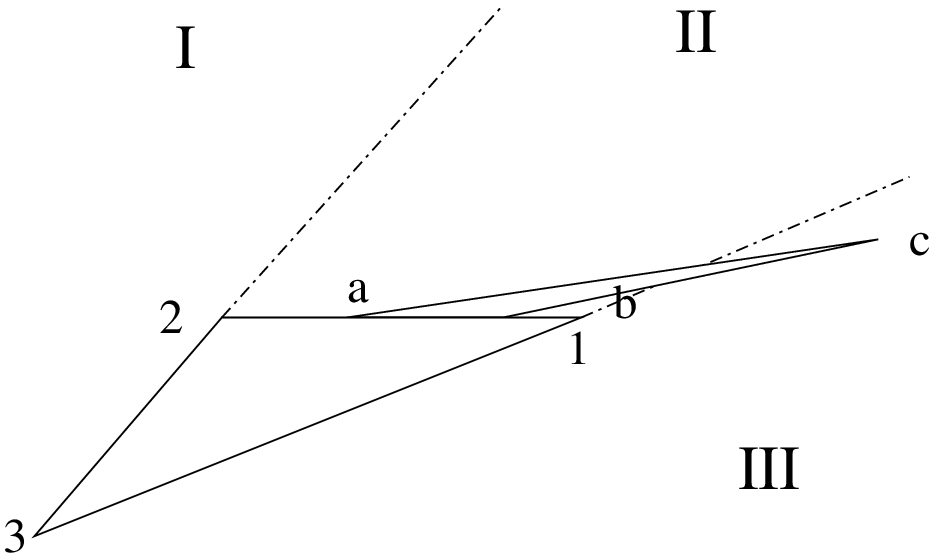} \\
{\bf (a)} & {\bf (b)} &  {\bf (c)}
\end{tabular}
\caption{\small\sf Touching triangles with one side contained in the other. {\bf (a)} Node $c$ in region I; {\bf (b)} Node $c$ in region II; {\bf (c)} Node $c$ in region III.}
\label{fig:inside}
\end{figure}

If point $c$ is placed in region II, both  $\vec{bc}$ and $\vec{ca}$ are to the left of
$\vec{23}$ and $\vec{31}$, so there are no more feasible angles, giving a total of two. 

If $c$ is in region III,
we get a new feasible angle formed by $\vec{31}$ and $\vec{bc}$. In this case, though, we are
left with $\vec{bc}$ and $\vec{ca}$ to the left of $\vec{23}$, and $\vec{31}$ to the left of
$\vec{ca}$. Thus, we have at most three feasible points. We also note that any triangle
associated with the feasible angle formed by $\vec{12}$ and $\vec{ca}$ cannot share an edge
with any triangle of the other two feasible angles, so there can be at most one edge among the
neighbors of $u$ and $v$. 

The argument is similar if $c$ is in region I. 

If points $1$ and $b$ are identical, the same arguments hold except, in addition, we no longer
have a feasible angle formed by $\vec{12}$ and $\vec{bc}$ because $\vec{12}$ is to the left
of $\vec{bc}$. Thus, we have at most two mutual neighbors and no edge between them. If points
$2$ and $a$ are the same, the same arguments hold. Putting these two cases together, we find that 
if $1$ and $b$ are identical and $2$ and $a$ are identical, there can be at most one feasible angle.

The remaining case occurs when neither shared side is contained in the other. This is the
situation represented by Figure~\ref{fig:overlap}.

\begin{figure}[ht]
\centering
\begin{tabular}{ccc}
\includegraphics[scale=0.42]{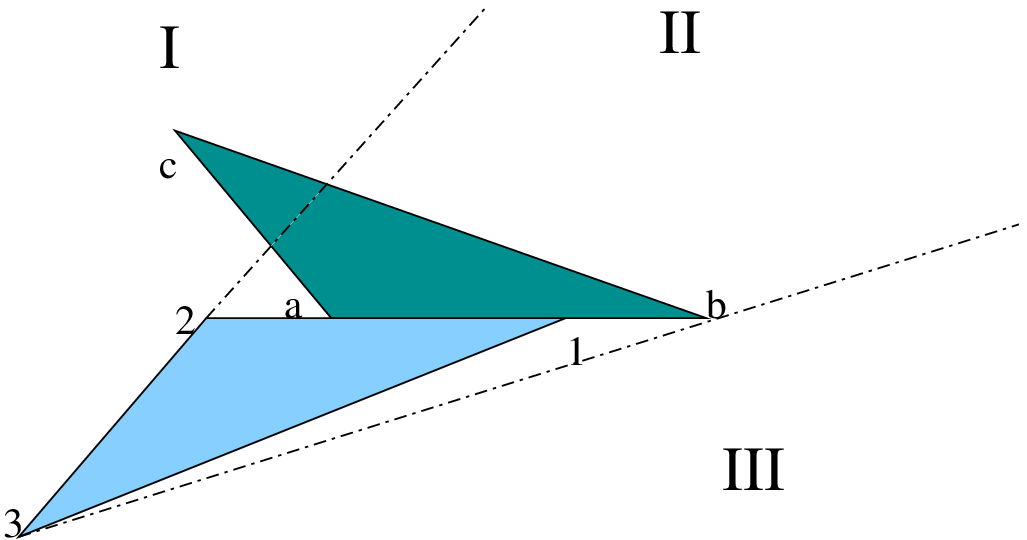} & \includegraphics[scale=0.42]{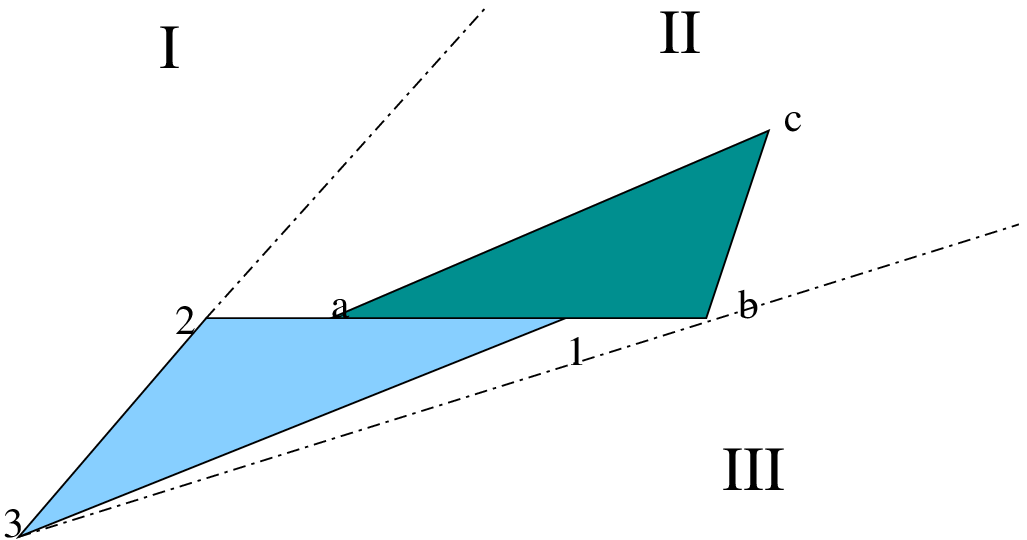} & \includegraphics[scale=0.42]{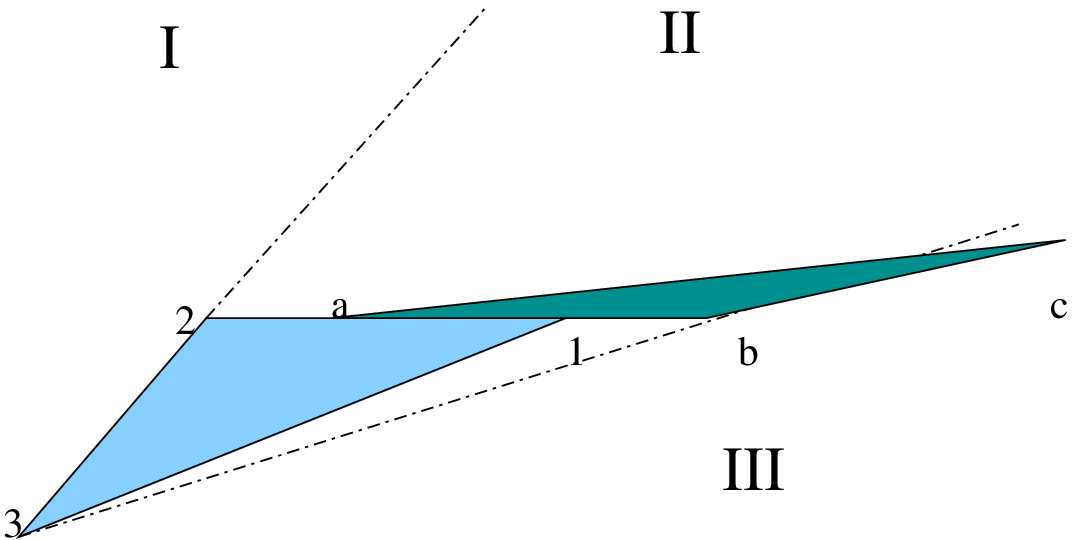} \\
{\bf (a)} & {\bf (b)} &  {\bf (c)}
\end{tabular}
\caption{\small\sf Touching triangles with touching sides overlapping. {\bf (a)} Node $c$ in region I; {\bf (b)} Node $c$ in region II; {\bf (c)} Node $c$ in region III.}
\label{fig:overlap}
\end{figure}

As previously, there can be no feasible angle associated with $\vec{12}$ and
$\vec{ab}$, but now we have feasible angles formed by $\vec{12}$ and $\vec{ca}$, and by
$\vec{31}$ and $\vec{ab}$. In addition, $\vec{12}$ is to the left of $\vec{bc}$ and
$\vec{ab}$ is to the left of $\vec{23}$. 
Again, we are reduced to considering the four
pairings of $\vec{23}$ and $\vec{31}$ with $\vec{bc}$ and $\vec{ca}$.
If  $\vec{ca}$ is to the right of $\vec{31}$, then
$\vec{31}$ is to the left of $\vec{ca}$, and vice versa, so that pairing is not possible.
Finally, we note that if $c$ is in regions I or II, then $\vec{23}$ and $\vec{31}$ are to the
left of $\vec{bc}$, while if $c$ is in regions II or III, $\vec{bc}$ and $\vec{ca}$ are to
the left of $\vec{23}$. So, if $c$ is in region II, there are at most two feasible angles. 
Otherwise, there can be three but, as above, at most two of the
associated triangles can touch. 
\end{proof}

With this theorem, we see that the left two graphs in the top 
row of Figure~\ref{fig-tough} are not in $TTG$.
We next consider what happens to the set of common neighbors 
if we relax the condition that there is an edge between two nodes.

\begin{theorem}\label{thm:nonside}
Let G be a $TTG$, and let $u$ and $v$ be any two nodes
in G. Then $|N_{uv}| \le 4$ and $|E_{uv}| \le 2$.
\end{theorem}
\begin{proof}
The proof follows that style of the previous theorem. 
Let $T_u$ and $T_v$ be the two triangles corresponding to nodes $u$ and $v$.
We have already dealt with the two triangles sharing a side above.
So, we then consider the case when a pair of sides are collinear, as illustrated
in Figure~\ref{fig:collinear}.

\begin{figure}[ht]
\centering
\vspace{-1cm}
\begin{tabular}{ccc}
\includegraphics[scale=0.5]{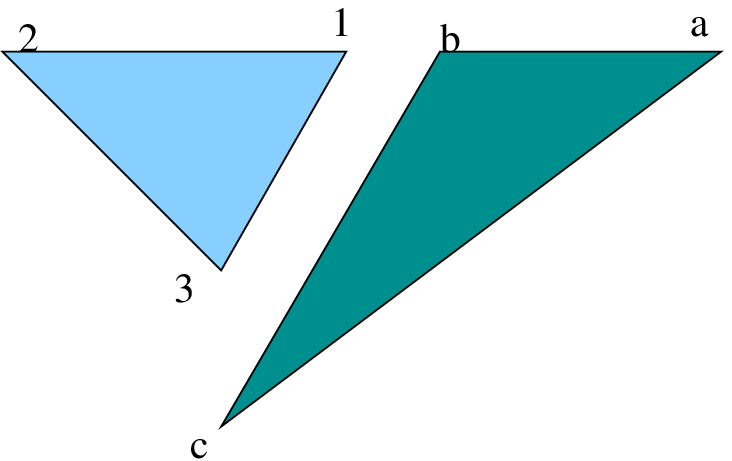} & \includegraphics[scale=0.5]{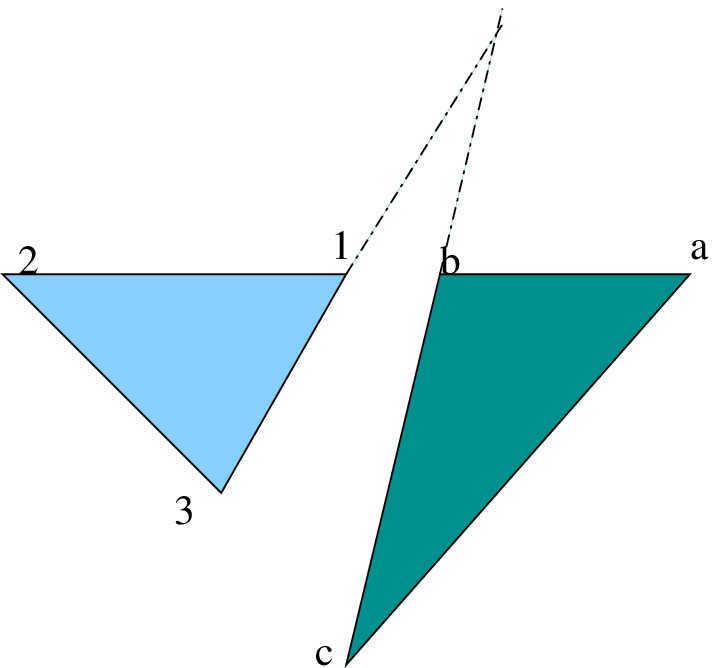} & \includegraphics[scale=0.5]{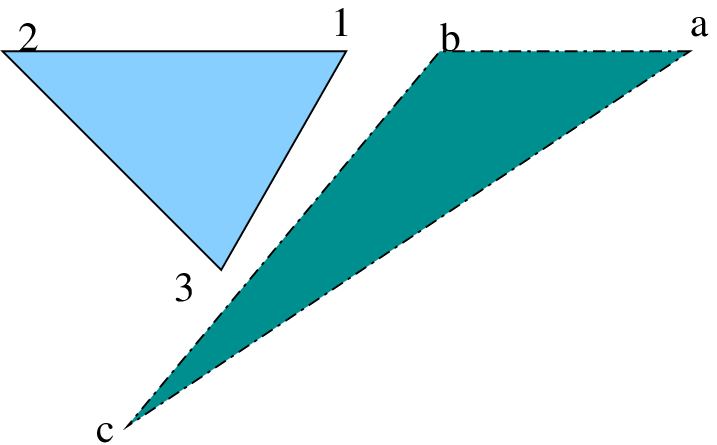} \\
{\bf (a)} & {\bf (b)} &  {\bf (c)}
\end{tabular}
\caption{\small\sf Collinear sides{\bf (a)} $\vec{31}$ and $\vec{bc}$ parallel; {\bf (b)} $\vec{31}$ and $\vec{bc}$ meeting above; {\bf (c)} $\vec{31}$ and $\vec{bc}$ meeting below.}
\label{fig:collinear}
\end{figure}

For this case, we can place a triangle touching $\vec{12}$ and $\vec{ab}$.
Since both triangles are to the left of both $\vec{12}$ and $\vec{ab}$, these
sides cannot be used in any other feasible angle. There can be no feasible
angle formed by $\vec{23}$ and $\vec{ca}$, since, if any part of $\vec{ca}$
is to the right of $\vec{23}$, the latter must be to the left of $\vec{ca}$,
and vice versa. In addition,
there can only be one of the two possible feasible angles 
formed by $\vec{23}$ and $\vec{bc}$ or by $\vec{31}$ and $\vec{ca}$.
Thus, there can be at most three touching triangles. 
(A more careful analysis shows that case (a) can have at most two,
while cases (b) and (c) will have three only if the triangles touch.)

For the next case, we consider when a vertex of one triangle
touches the interior of a side of the other, as shown in
Figure~\ref{fig:pointside}. The dotted lines indicate the lines 
$\vec{23}$ and $\vec{31}$,
and divide the area into three regions.
We consider the cases determined by which regions contain vertices $a$ and $b$. We note that if
$a$ is in region I, $b$ must also be in that region. We can also assume that both $a$ and $b$ do
not lie on either $\vec{23}$ and $\vec{31}$, as this was covered by the collinear case addressed above.
In all cases, we have feasible points determined by
$\vec{12}$ with both $\vec{ca}$ and $\vec{bc}$. 
Also, in all cases either $\vec{ca}$ is to the left of $\vec{23}$, or vice versa, so this pair
is eliminated. The similar condition holds for $\vec{bc}$ and $\vec{31}$.

For the case when $\vec{ab}$ lies in region II (Figure~\ref{fig:pointside}(a)), 
$\vec{12}$ can also form a feasible point with $\vec{ab}$. On the other
hand, the triangle $abc$ lies to the left of both $\vec{23}$ and $\vec{31}$, 
so we are limited to three feasible points. 

When $\vec{ab}$ lies in region I (Figure~\ref{fig:pointside}(b)), the 
triangle $abc$ is to the left of $\vec{31}$, so the latter has no
feasible points. There is always a feasible point fixed by
$\vec{23}$ and $\vec{bc}$. If $\vec{12}$ is to the left of
$\vec{ab}$, the only remaining possibility is given 
by $\vec{ab}$ and $\vec{23}$.
If $\vec{12}$ is partly to the right of $\vec{ab}$,
both  $\vec{ab}$ and $\vec{23}$ and  $\vec{ab}$ and $\vec{12}$
give feasible points, but a triangle placed at one blocks the other
(and the feasible point of $\vec{23}$ and $\vec{bc}$ as well). Thus,
we are limited to four touching triangles.

The case when $\vec{ab}$ lies in region III (Figure~\ref{fig:pointside}(c))
is symmetric.

We next consider $b$ in region I and $a$ in 
region II (Figure~\ref{fig:pointside}(d)). The
triangle $abc$ is to the left of $\vec{31}$, so the latter has no
feasible points. In addition, $\vec{23}$ is to the left of
$\vec{ab}$, leaving at most four feasible points.

If we leave $a$ in region II but move $b$ to region III
(Figure~\ref{fig:pointside}(e)), we have a similar situation,
with triangle $abc$ is to the left of $\vec{23}$
and $\vec{ca}$ is to the left of $\vec{31}$.
 
Switching their roles, with $a$ in region III and $b$ in region II
(Figure~\ref{fig:pointside}(f)),
we still have  triangle $abc$ is to the left of $\vec{23}$ but now
$\vec{31}$ is to the left of $\vec{ab}$.

In the final sub-case, $b$ lies in region I and $a$ lies in region III
(Figure~\ref{fig:pointside}(g)).
Here, the triangle $123$ lies to the left of $\vec{ab}$, eliminating
all feasible points involving the latter. We are left with two
remaining possibilities: $\vec{bc}$ with $\vec{23}$ 
and $\vec{ca}$ with $\vec{31}$, for a total of four.

\begin{figure}[th]
\centering
\vspace{-1cm}
\begin{tabular}{ccc}
\includegraphics[scale=0.5]{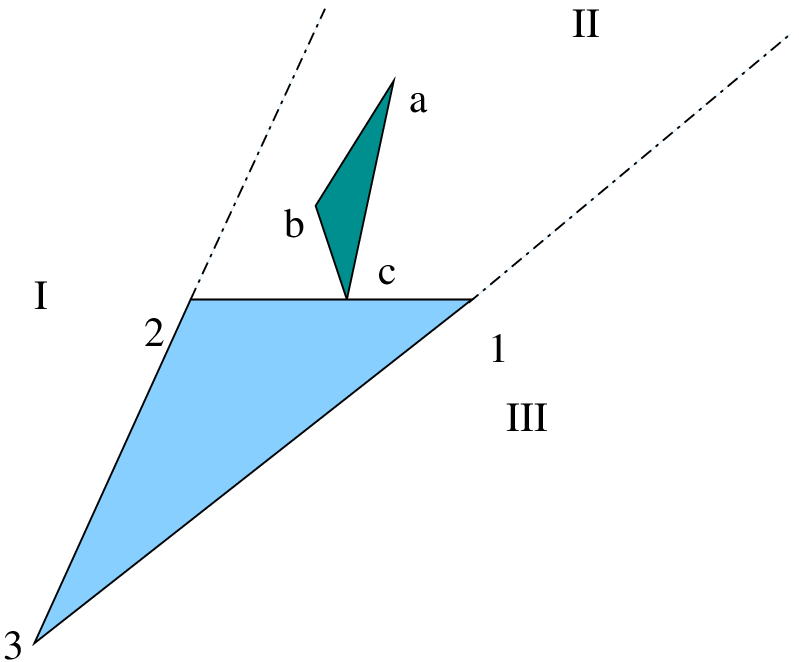} & \includegraphics[scale=0.5]{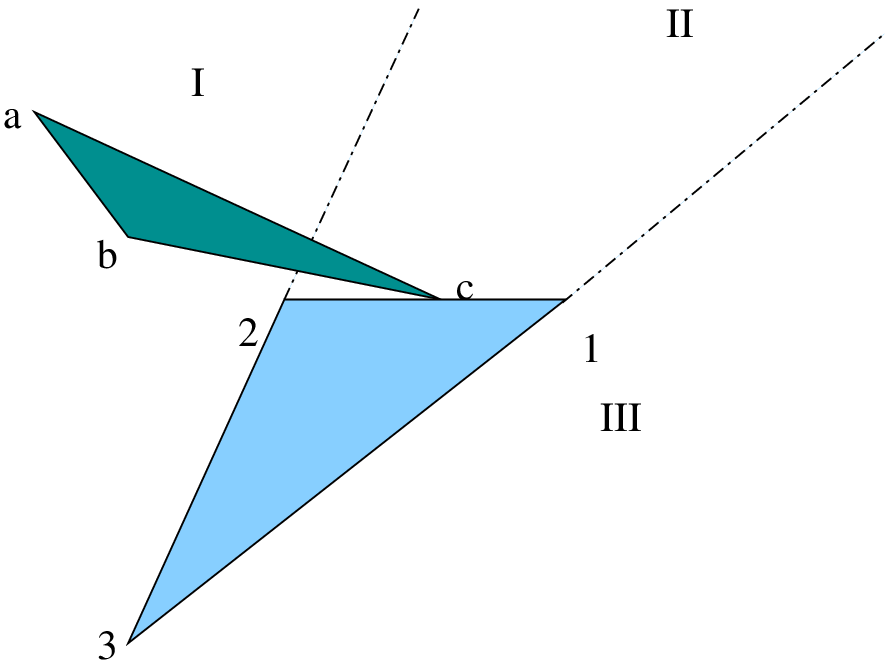} & \includegraphics[scale=0.5]{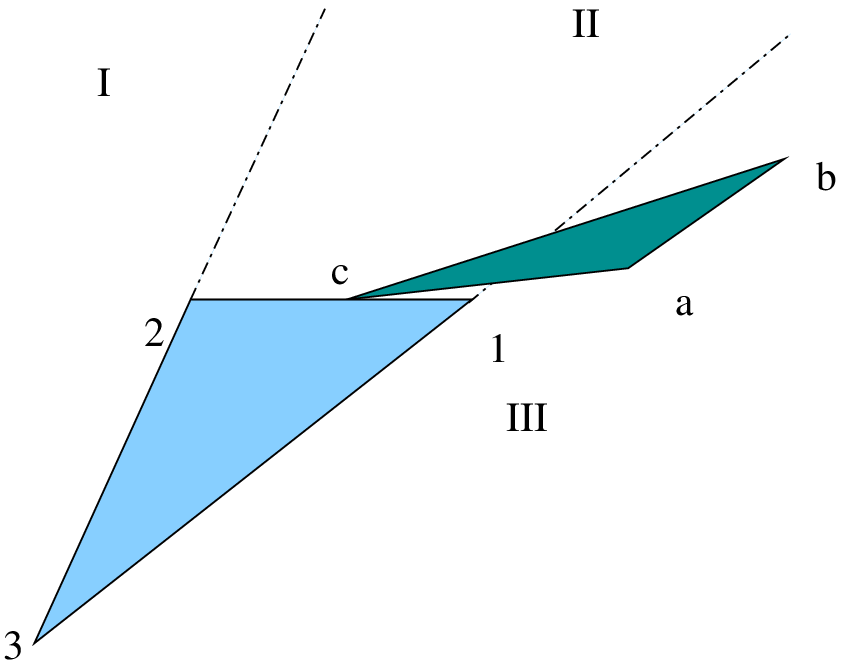} \\
{\bf (a)} & {\bf (b)} & {\bf (c)} \\
\includegraphics[scale=0.5]{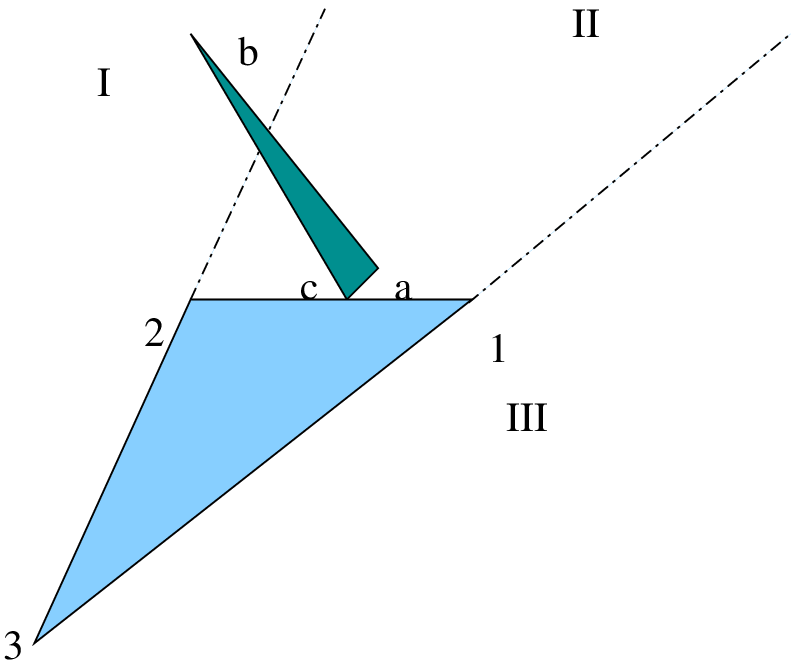} & \includegraphics[scale=0.5]{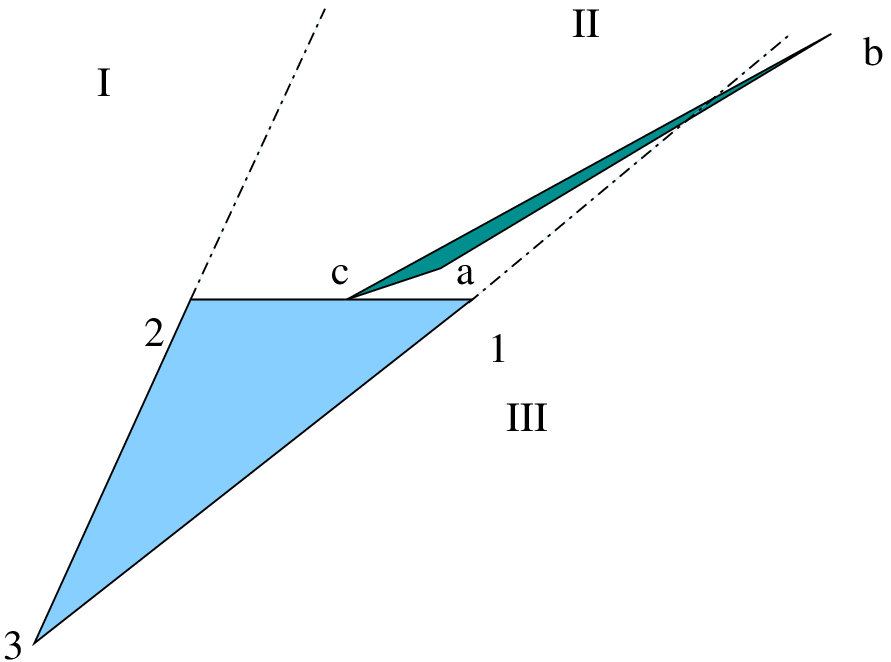} & \includegraphics[scale=0.5]{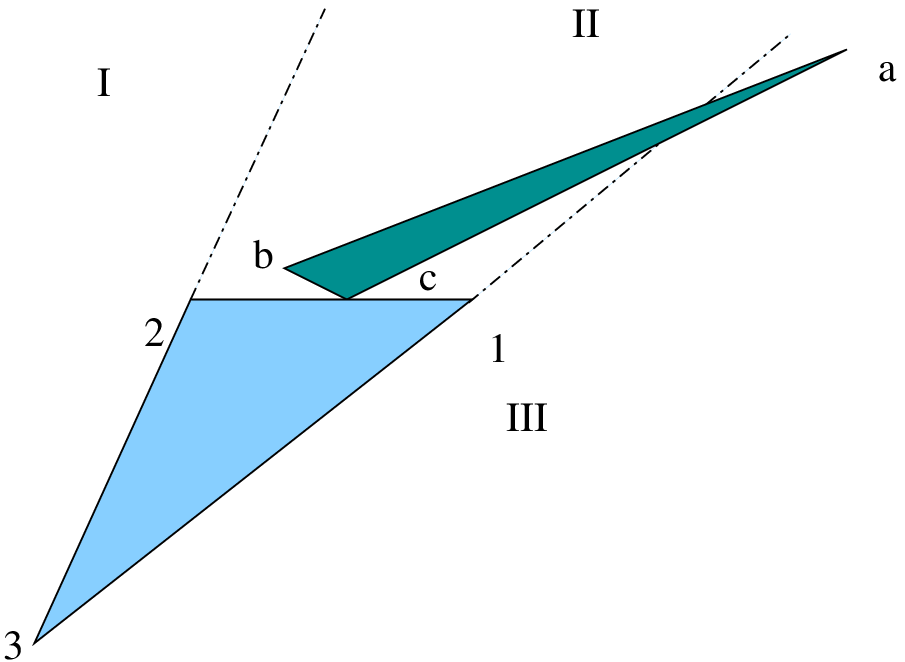} \\
{\bf (d)} & {\bf (e)} & {\bf (f)} \\
 & \includegraphics[scale=0.5]{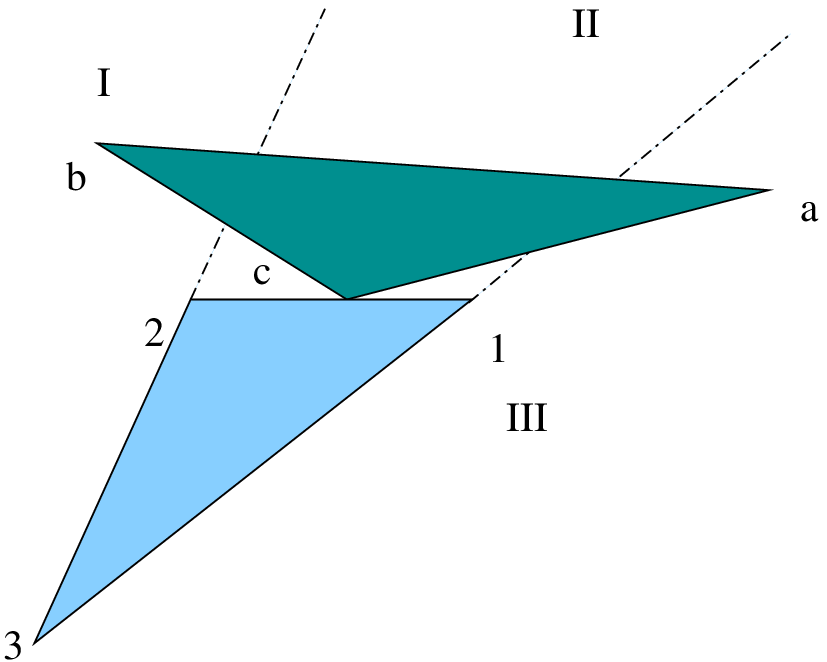} & \\
 & {\bf (g)} & \\
\end{tabular}
\caption{\small\sf Triangles touch at vertex and side. {\bf (a)} $\vec{ab}$ to region II; {\bf (b)} $\vec{ab}$ in region I; {\bf (c)} $\vec{ab}$ in region III; {\bf (d)} $b$ in region I, $a$ in region II; {\bf (e)} $b$ in region III, $a$ in region II; {\bf (f)} $b$ in region II, $a$ in region III; {\bf (g)} $b$ in region I, $a$ in region III.}
\label{fig:pointside}
\end{figure}

Next, we assume the triangles touch at two vertices, as shown in
Figure~\ref{fig:pointpoint}. There can be a feasible point formed by $\vec{23}$ and $\vec{ca}$,
and one by $\vec{31}$ and $\vec{bc}$. On the other hand, we can immediately eliminate the pairs
 $\vec{23}$ and $\vec{bc}$, and $\vec{31}$ and $\vec{ca}$. If $\vec{ab}$ is in the left half
plane of $\vec{12}$ (Figure~\ref{fig:pointpoint}(a)), the 
latter has no feasible points. Thus, there can be at most four. In fact,
$\vec{ab}$ can have at most one feasible point, with either $\vec{31}$ or  $\vec{23}$, but not
both, so there are at most 3 feasible points.

Otherwise, either point $a$ or point $b$ is to the right of $\vec{12}$ (Figure~\ref{fig:pointpoint}(b)),
all of triangle $123$ is to the left of $\vec{ab}$, and the symmetric case holds, with no
feasible points associated with $\vec{ab}$, and at most one additional feasible point formed by
$\vec{12}$ and either $\vec{ca}$ or $\vec{bc}$.

Finally, if the triangles do not touch at all and do not have a
pair of collinear sides, consider a pair of
closest points $p_0$ and $p_1$, one on each triangle, and the line segment
between the two points. If we imagine translating the points along
this line segment until the triangles touch, we have one of the
three situations: that of Theorem~\ref{thm:sideside}, 
Figure~\ref{fig:pointside} or Figure~\ref{fig:pointpoint}, 
and similar analysis apply, but with a possible reduction in
usable feasible points. For example, consider the configuration 
of Figure~\ref{fig:pointside}(c).


This fits the pattern of 
Figure~\ref{fig:pointpoint}(a). Thus, $\vec{12}$ has no feasible points,
and $\vec{ab}$ might potentially form a feasible point with $\vec{23}$
or $\vec{31}$, but not both. Now, unlike the touching case, 
we have four feasible points from sides $\vec{23}$, $\vec{31}$, $\vec{bc}$
and $\vec{ca}$. The problem is that, if a triangle is placed at
one of those points, the remainder become unusable. Thus, we end up
with at most three neighboring triangles. To complete the proof, 
we note that, in all of the cases, there can be at most two pairs of
touching triangles among the ones added.
\end{proof}

\begin{figure}[th]
\centering
\vspace{-1cm}
\begin{tabular}{ccc}
\includegraphics[scale=0.5]{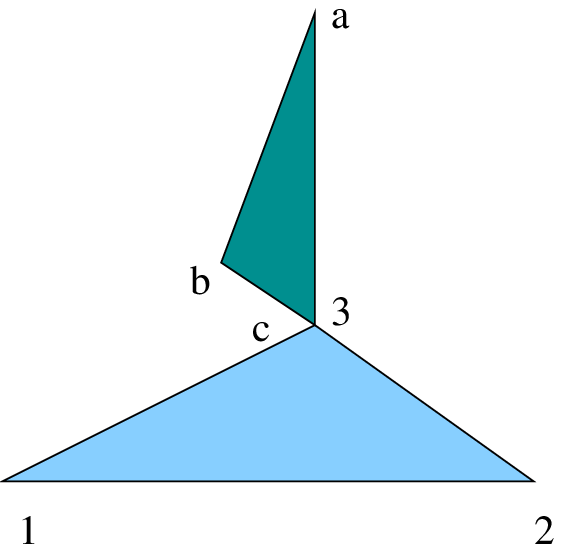}
& \includegraphics[scale=0.5]{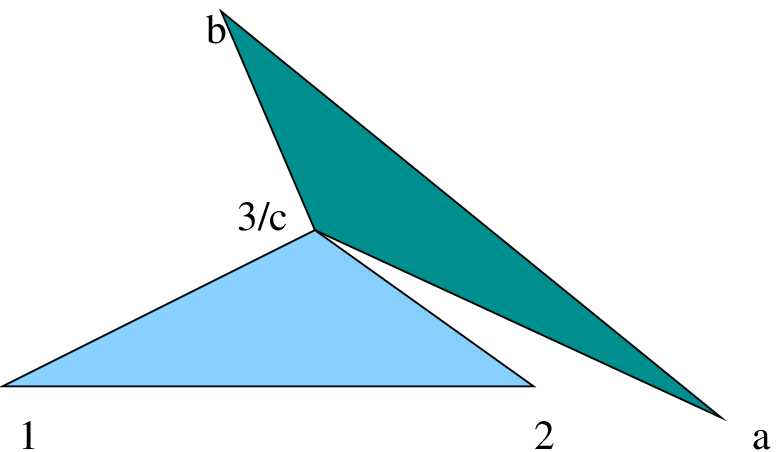} & \includegraphics[scale=0.5]{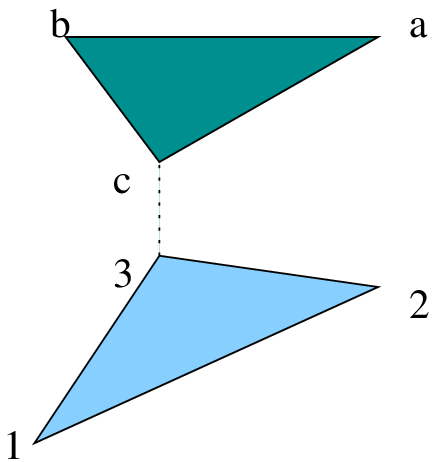}\\
{\bf (a)} & {\bf (b)} & {\bf (c)}
\end{tabular}
\caption{\small\sf Triangles touch at two vertices {\bf (a)} $\vec{ab}$ to the
  left of $\vec{12}$; {\bf (b)} $\vec{ab}$ crossing $\vec{12}$; {\bf
    (c)} shows one case of non-touching triangles.}
\label{fig:pointpoint}
\end{figure}

Figures~\ref{fig:pointside} and~\ref{fig:pointpoint} show that the
bounds of 3 or 4 derived in the proof are tight. Theorem~\ref{thm:nonside} shows that the top right graph in
Figure~\ref{fig-tough} is not $TTG$. Although these two theorems
provide simple tests for eliminating potential $TTG$s, we are fairly
certain that they do not provide sufficient conditions.

\section{Conclusion and Future Work}
We have considered the class of graphs that can be represented as
contact graphs of triangles, and shown that this includes outerplanar
graphs as well as subgraphs of square and hexagonal grids.
We derived some necessary conditions for such graphs, and was able to
present a complete characterization of the special subclass of biconnected
triangulation graphs.
A complete characterization of $TTG$, as well as contact graphs of
4-gons and 5-gons, remains open.
%
%


\bibliographystyle{abbrv}
{
\begin{small}
\vspace{-.3cm}\bibliography{stephen}
\end{small}
}
\newpage

\section*{Appendix: Non-$TTG$ Planar Graphs}

Here we briefly illustrate that the property ``representable as $TTG$''
is not closed under homeomorphisms or minors. Specifically, the graphs
in Figure~\ref{fig-tough} cannot be represented as $TTG$s, but
subdividing one edge from each of them makes them representable as $TTG$s. 

\begin{figure}[th]
\begin{center}
\includegraphics[width=13cm]{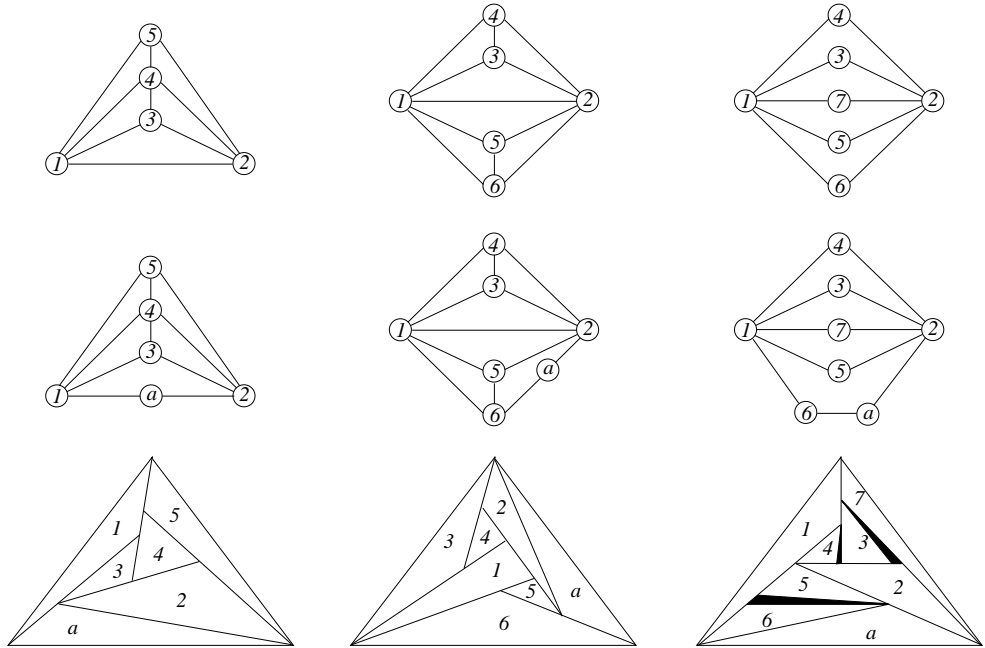}
\caption{\small\sf The graphs in the first row (on 5, 6, 7 vertices) do not have touching triangle graph representations. However, subdividing one edge from each, as in the second row, results in graphs that have $TTG$ representations. These representations are shown in the third row.
\label{fig-tough}}
\end{center}
\end{figure}

\end{document}